\documentclass{article}

\usepackage[utf8]{inputenc}
\usepackage{amsmath,amssymb
,amsthm
}
\usepackage[english]{babel}
\usepackage[ruled,vlined,titlenumbered,linesnumbered,longend]{algorithm2e}
\usepackage{tikz}
\usetikzlibrary{trees,snakes,shapes}
\usepackage{walk}

\DeclareMathOperator{\MSet}{\textsc{MSet}}

\def\cM{\mathcal{M}}

\def\cA{\mathcal{A}}

\def\eps{\varepsilon}

\newtheorem{theorem}{Theorem}[section]
\newtheorem{lemma}[theorem]{Lemma}
\newtheorem{corollary}[theorem]{Corollary}
\newtheorem{proposition}[theorem]{Proposition}

\newtheorem{definition}[theorem]{Definition}
\newtheorem{remark}[theorem]{Remark}

\title{Asymptotic Analysis and Random Sampling of Digitally Convex Polyominoes}
\author{O. Bodini%
\thanks{LIPN, Université Paris 13, 99, avenue Jean-Baptiste Clément 93430
Villetaneuse, France}
\and Ph. Duchon%
\thanks{LaBRI, 351, cours de la Libération F-33405 Talence cedex, France}
\and A. Jacquot%
\thanks{LIPN, Université Paris 13, 99, avenue Jean-Baptiste Clément 
93430 Villetaneuse, France}
\and L. Mutafchiev%
\thanks{Institute of Mathematics and Informatics Acad. Georgi Bonchev Str.,
Block 8 1113 Sofia, BULGARIA}
}

\begin{document}
\maketitle

\begin{abstract}
Recent work of Brlek \textit{et al.} gives a characterization of digitally
convex polyominoes using combinatorics on words. From this work,
we derive a combinatorial symbolic description of digitally convex polyominoes
and use it to analyze their limit properties and  build a uniform
sampler. Experimentally,
our sampler shows a limit shape for large digitally convex
polyominoes.
\end{abstract}

\section*{Introduction}


In discrete geometry, a finite set of unit square cells is said to be
a \emph{digitally convex polyomino}\footnote{The usual definition of a
  polyomino requires the set to be connected, whereas digitally convex
  sets may be disconnected. 
However, one can coherently define some
  polygon as the boundary of any digitally convex polyomino; in the case
  of disconnected sets, this boundary will not be self-avoiding.

} if
it is exactly the set of unit cells included in a convex region of the
plane. We only consider
digitally convex polyominoes up to translation.
The \emph{perimeter} of a digitally convex polyomino is that of
the smallest rectangular box that contains it.

The notion of digitally convex polyominoes arises naturally in the context of curvature estimators and consequently in the important field of pattern recognition \cite{KLN2008,CMT2001}.

Brlek \emph{et al.} \cite{Brlek09a} described a characterization of digitally convex polyominoes,
in terms of words coding their contour.
In this paper, we reformulate this characterization in the context of constructible combinatorial classes 
and we use it to build and analyze an algorithm to randomly sample digitally convex polyominoes.

Our algorithm, based on a model of parametrized samplers, called Boltzmann samplers \cite{DuFlLoSc04}, draws
digitally convex polyominoes at random. Although all
possible digitally convex polyominoes have positive probability,
the perimeter of the randomly generated polyomino being itself
random, two different structures with the same perimeter
appear with the same probability. Moreover, an appropriate choice of the
tuning parameter allows the user to adjust the random model,
typically in order to generate \emph{large} structures. 
We present also in this paper how to tune this parameter. 

The Boltzmann model of random sampling, as introduced
in~\cite{DuFlLoSc04}, is a general method for the random generation of
discrete combinatorial structures where, for some real parameter
$x>0$, each possible structure $\gamma$, with (integer) size
$|\gamma|$, is obtained with probability proportional to
$x^{|\gamma|}$.
A \emph{Boltzmann sampler} for a
combinatorial class $\mathcal{C}$ is a randomized algorithm that takes
as input a parameter $x$, and outputs a random element of
$\mathcal{C}$ according to the Boltzmann distribution with parameter $x$.

Samples obtained via this generator suggested that large random quarters of digitally convex polyominoes exhibit a limit shape. 
We identify and prove this limit shape in Section \ref{limitshape}

The first
section is dedicated to introducing the characterization of Brlek
\emph{et al} \cite{Brlek09a} in the framework of symbolic methods. In Section \ref{limitshape},
we analyze asymptotic properties of quarters of digitally convex polyominoes. Finally, we give 
 in Section \ref{section3} the samplers for digitally convex polyominoes and
some analysis for the complexity of the sampling. 

\section{Characterization of digitally convex polygons}
The goal of this section is to recall (without proofs) the characterization by Brlek, Lachaud,
Proven\c{c}al, Reutenauer \cite{Brlek09a} of digitally convex
polyominoes and recast it in terms of the symbolic method. This characterization is the starting point to
efficiently sample large digitally convex polyominoes,
and is thus needed in the next chapters. 


\subsection{Digitally convex polyominoes}
%
%
\begin{definition}
A \emph{digitally convex polyomino}, DCP for short, is the set of all cells of $\mathbb{Z}^2$
included in a bounded convex region of the plane.
\end{definition}

\begin{figure}
\begin{center}
\includegraphics[scale=1]{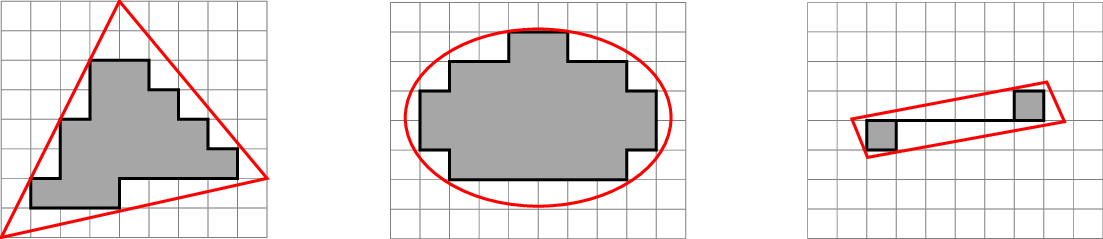}
\end{center}
\caption{A few digitally convex polyominoes (in grey) of perimeter $24$,$26$ and $16$, and their contour (in black) }
\label{polex}
\end{figure}
A first geometrical characterization directly follows from the definition: a set of cells of the
square lattice $P$ is a digitally convex polyomino if all cells included in the convex hull of $P$
is in $P$.


For our propose, a DCP will be rather 
 characterized through its \textit{contour}.
\begin{definition}
The \emph{contour} of a polyomino $P$ is the closed non-crossing path on the square lattice with no half turn allowed such that its interior is $P$.
In the case where $P$ is not connected, we take as the contour, the only such path which stays inside the convex hull of $P$.
\end{definition}

We define the \emph{perimeter} of $P$ to be the length of the
contour (note that, for digitally convex poyominoes, it is equal to the perimeter of the
smallest rectangular box that contains $P$).

The contour of DCP can be decomposed into four specifiable sub-paths through the standard decomposition of polyominoes.


\medskip

The standard decomposition of a polyomino distinguishes four extremal points:
\begin{itemize}
\item{$W$ is the lowest point on the leftmost side}
\item{$N$ is the leftmost point in the top side}
\item{$E$ is the highest point on the rightmost side}
\item{$S$ is the rightmost point on the bottom side}
\end{itemize}
\begin{figure}
\includegraphics[scale=1]{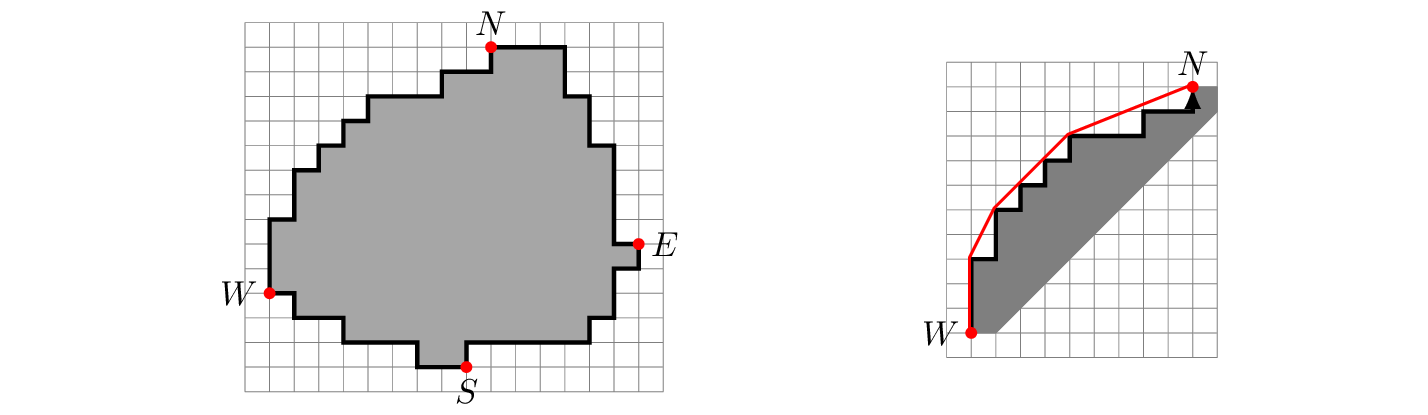}
\caption{A polyomino composed of $4$ NW-convex. The $W$ to $N$ NW-convex path is coded by $w=1110110101010001001$.}
\label{NESW}
\end{figure}

The contour of a DCP is then the union of the four (clockwise) paths
$WN$, $NE$, $ES$ and $SW$. Rotating the latter three paths by,
respectively,
a quarter turn, a half turn, three quarter turn counterclockwise leaves
all paths containing only north and east steps; digital convexity is
characterized by the fact that each (rotated) side is
\emph{NW-convex}.

\begin{definition}
A path $p$ is \emph{NW-convex} if it begins and ends by a north step and
if there is no cell between $p$ and its upper convex hull (see Fig.~\ref{NESW}).
\end{definition}

In the following, we mainly focus on the characterization and random sampling of NW-convex paths.

\subsection{Words}
The characterization in \cite{Brlek09a} is based on
combinatorics on words. Let us recall some classical notations.
We are interested in words on the alphabet $\{0,1\}$.  Thus
$\{0,1\}^{*}$ is the set of all words, $\{0,1\}^{+}$ is the set of
non-empty words. 

Each NW-path can be bijectively coded by a word $w\in\{0,1\}^*$. From $W$ to $N$, the letter $0$ encodes a horizontal (east) step and $1$ a
vertical (north) step (see figure \ref{NESW}).

The idea is to decompose a NW-convex path $w$ by contacts between $w$ and its convex hull.
%
%
%

\begin{definition}
Let $p$, $q$ be two integers, with $p\geq 0$ and  $q> 0$. The Christoffel word associated to $p$, $q$ is the word which codes the highest path going from $(0,0)$ to $(p,q)$ while staying under the line going from $(0,0)$ to $(p,q)$. 
A Christoffel word is primitive if $p$ and $q$ are coprime. 
\end{definition}
\vspace{-0.6cm}
\begin{figure}
\begin{center}
\includegraphics[scale=1]{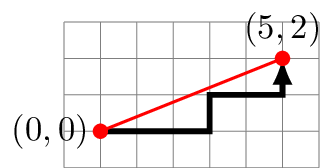}
%
\end{center}
\caption{The Christoffel primitive word 0001001, of slope $2/5$.}
\label{dessinchrist}
\end{figure}
\vspace{-0.4cm}
Note that a Christoffel primitive word always ends with $1$. 
%

\subsection{Symbolic characterization of NW-convex paths}

Let us recall in this section, two basic notions in analytic combinatorics: the combinatorial classes and the enumerative generating functions.

A \emph{combinatorial class} is a finite or countable set
$\mathcal{C}$, together with a \emph{size} function $|.|:
\mathcal{C}\mapsto \mathbb{N}$ such that, for all $n\in\mathbb{N}$,
only a finite number $c_n$ of elements of $\mathcal{C}$ have size $n$.
The (ordinary) \emph{generating function} $C(z)$ for the class
$\mathcal{C}$ is the formal power series $C(z) = \sum_{n} c_n z^n =
\sum_{\gamma\in\mathcal{C}} z^{|\gamma|}$. If $C(z)$ has a positive
(possibly infinite) convergence radius $\rho$ (which is equivalent to
the condition that $\overline{\lim}\, c_n^{1/n} < \infty$), standard
theorems in analysis imply that the power series $C(z)$ converges and
defines an analytic function in the complex domain $|z|<\rho$; here we
will only use the fact that $C(z)$ is defined for real $0<z<\rho$.


Our sampler is based on the following decomposition theorem. 

\begin{theorem}\cite{Brlek09a}
A word is NW-convex if and only if it is a sequence of Christoffel primitive words of decreasing slope, beginning with $1$.
\end{theorem}

The reason of a NW-convex path to begin with a vertical step is to
avoid half-turn on the contour of a polyomino. Indeed, since all
Christoffel primitive words end with $1$, this ensures compatibility
with the standard decomposition (beginning and ending on a
corner).  Since the Christoffel primitive words appear in
decreasing order, NW-convex paths can be identified with multisets
of Christoffel primitive words, with the condition that the word
$1$ appears at least once in the multiset (this condition can be
removed by removing the initial vertical step from NW-convex
paths).  This is the description we will use in what follows
next.

The generating function of Christoffel primitive words, counted by
their lengths, is $\sum\limits_{n\geq1}\varphi(n)z^n$, with
$\varphi$ the Euler's totient function. It follows \cite[p.~29]{Flajolet2009}
that the
generating function of the class $\mathcal{S}$ of NW-convex paths,
counted by length, is $S(z)=\prod\limits_{n=1}^{\infty}
(1-z^{n})^{-\varphi(n)}$. More precisely, we can use the 3-variate generating function 
$S(z,h,v)=(1-zv)^{-1}\prod\limits_{n=2}^{\infty}\prod\limits_{p+q=n, p\wedge q=1}
(1-z^{n}v^ph^q)^{-1}$, to describe by $[z^nh^iv^j]S(z,h,v)$ the number of NW-convex paths beginning in $(0,0)$ and terminating in position $(i,j)$.


\section{Asymptotics for NW-convex paths and its limit shape}
\label{limitshape}

This section is dedicated to the analysis of some properties of NW-convex paths. 
The main objective is to describe a limit shape for the normalized random NW-convex paths. This is obtained in three steps. 
In the first one we extract the asymptotic of NW-convex paths using a Mellin transform approach. 
In the second one, using the same approach we prove that the asymptotic of the average number of initial vertical steps of a NW-convex path is
 in $O(\sqrt[3]{n})$. Then, using some technical lemmas, we conclude with the fact that the limit shape
 is $\sqrt{2z}-z$ with $0\leq z \leq 1/2.$ Let us begin by a brief overview on Mellin transforms. For more details, see \cite{Flajolet2009}.

\subsection{Brief overview on Mellin  transforms}
\label{MellinOverview}
The Mellin transform is an integral transform similar to Laplace
transform from which we can derive asymptotic estimates of
expressions involving specific infinite products or sums.  Given a
continuous function $f$ defined on $\mathbb{R}^+$, the
\emph{Mellin transform} of $f$ is the function
\begin{equation}\label{eq:def_mellin}
\cM[f](s):=\int_0^{\infty}f(t)t^{s-1}\mathrm{d}t.
\end{equation}

If $f(t)=O(t^{-a})$ as $t\to 0^+$ and $f(t)=O(t^{-b})$ as
$t\to+\infty$, then $\cM[f](s)$ is an analytic function defined on the
\emph{fundamental strip} $a<\mathrm{Re}(s)<b$. In addition, $\cM[f](s)$
is in most cases continuable to a meromorphic function in the whole
complex plane.

The first fundamental property of the Mellin transform is that
it \emph{factorizes} harmonic sums as follows:
\begin{equation}
\label{eq:fac}
G(t)=\sum_{k\geq 1}a_kf(\mu_kt)\ \Rightarrow\ \cM[G](s)=\Big(\sum_{k\geq 1}a_k\mu_k^{-s}\Big) \cM[f](s).
\end{equation}

In a similar way as the Laplace transform, the Mellin transform is almost involutive,
the function $f(t)$ being recovered from $\cM[f](s)$ using the
\emph{inversion formula}
\begin{equation}\label{eq:inverse_mellin}
f(t)=\int_{c-i\infty}^{c+i\infty}\cM[f](s)t^{-s}\mathrm{d}s\ \ \ \ \mathrm{for\ any\ }c\in(a,b).
\end{equation}
From the inversion formula and the residue theorem, the asymptotic expansion of $f(t)$ as $t\to 0^-$
can be derived from the poles of $\cM[f](s)$ on the left of the fundamental
domain, the rightmost such pole giving the dominant term of the asymptotic
expansion. If $\cM[f](s)$ is decreasing very fast as
$\mathrm{Im}(s)\to\infty$
 (which occurs in all series to be analyzed next, based on the
fact that $\Gamma (s)$ is decaying fast and $\zeta(s)$ is of
moderate growth as $\mathrm{Im}(s)\to\infty$), then the
following \emph{transfer rule} holds: a pole of $\cM[f](s)$ of
order $k\!+\!1$ ($k\geq 0$),
$$
\cM[f](s)\mathop{\sim}_{s\to\alpha}\lambda_{\alpha}\frac{(-1)^kk!}{(s-\alpha)^{k+1}}
$$
yields a term
$$\lambda_{\alpha}t^{-\alpha}\ln(t)^k$$
in the singular expansion of $f(t)$ around $0$.
In particular, a simple pole $\lambda_{\alpha}/(s-\alpha)$ yields
a term $\lambda_{\alpha}/t^{\alpha}$.

The first step is to determine the Mellin transform of
$\ln(S(e^{-t}))$ which is easier than to obtain that the
Mellin transform of $S(z)$. After that it is quite easy to compute the expansion for $S(z)$ when $z$ tends to 1.
\begin{lemma}
\label{pouranne}
The Mellin transform associated with the
series of irreducible discrete segments is
  $$\cM[\ln(S(e^{-t}))](s)={\frac {\zeta  \left( s+1 \right) \zeta  \left( s-1 \right) \Gamma
 \left( s \right) }{\zeta  \left( s \right) }},$$
 where $\zeta(z)$ and $\Gamma(z)$ denote the Riemann zeta function and the Gamma function, respectively.
 \end{lemma}
 \begin{proof} 
 The proof relies on an exp-log schema rewriting $\ln(S(e^{-t}))$ as an harmonic sum, and then applying harmonic sum properties.
 Indeed, by an exp-log schema, we can rewrite $\ln(S(e^{-t}))$ as an harmonic sum. Indeed, we have :
 $$ \ln(S(e^{-t})) = \sum\limits_{n=1}^\infty -\varphi(n)\ln(1-e^{-tn}) $$
 and by an expansion of $\ln(1-x)$, we get :
  $$ \ln(S(e^{-t})) = \sum\limits_{n=1}^\infty \sum\limits_{k=1}^\infty\frac1{k}\varphi(n)e^{-tkn} $$
  We continue by swapping the sums
 $$ \ln(S(e^{-t})) = \sum\limits_{k=1}^\infty \frac1{k}B(kz)$$ with $B(z)=\sum\limits_{n=1}^\infty \varphi(n)e^{-nt}.$\\
 With the harmonic property of Mellin transform \ref{eq:fac}, we see that :
 $$ \cM[\ln(S(e^{-t}))](s) = \zeta(s+1)\cM[B](s)$$
 But, we can now apply the harmonic property to $B$ and we get $$\cM[B](s)=\dfrac{\zeta(s-1)}{\zeta(s)}\cM(e^{-t})(s)$$
And finally, as $\cM(e^{-t})(s)=\Gamma(s)$, we obtain :
  $$\cM[\ln(S(e^{-t}))](s)={\frac {\zeta  \left( s+1 \right) \zeta  \left( s-1 \right) \Gamma
 \left( s \right) }{\zeta  \left( s \right) }}$$

   \end{proof}

   \begin{proposition} \label{prop:equiv} We have the following equivalence for $S(z)$ when $z$ tends to 1:
  $$S(z) \sim \dfrac{\exp\left({\frac {6z\zeta  \left( 3 \right) }{{\pi }^{2} \left( 1-z \right) ^{2
}}}\right)}{ \left(2\pi(1-z)\right)^{\frac1{6}}}\cdot \exp\left(
{\dfrac {g(1-z)+\zeta  \left( 3 \right) }{2{\pi
}^{2}}}-2\,\zeta^\prime \left(-1 \right)\right ), $$ 
where
 $$ g(t)=\sum_r t^{-r}\Gamma(r)\zeta(r+1)\zeta(r-1) $$
and $r$ runs over the non-trivial zeros of the Riemann zeta
 function.
\begin{equation}\label{sumres}
 g(t)=\sum_r Res_{s=r}(F(s)),
 \end{equation}
  \begin{equation}\label{integrand}
  F(s) =\frac{t^{-s}\zeta(s+1)\zeta(s-1)\Gamma(s)}{\zeta(s)},
 \end{equation}
 where the right-hand side of (\ref{sumres}) represents the sum of the residuals
 of $F(s)$ over the non-trivial zeros of the Riemann zeta
 function.



\end{proposition}
\begin{proof}
 Our argument is similar to that presented by Brigham \cite{Brigham1950} and Yang \cite{Yang2000} in the study of partitions
 of integers into primes. First, we apply the inversion formula for Mellin
 transforms to the function from Lemma 1 and obtain
 $$
 \ln{S(e^{-t})}
 =\frac{1}{2\pi i}\int_{c_0-i\infty}^{c_0+i\infty}t^{-s}\Gamma(s)\zeta(s+1)
 \frac{\zeta(s-1)}{\zeta(s)}ds
 $$
 with $c_0>2$. As in \cite{Brigham1950,Yang2000}, we now shift the line of integration to the
 vertical $\Re{(s)}=-c_1, 0<c_1<1$, taking into account the residues
 of the integrand at $s=2,0$ and at the non-trivial zeros of the zeta
 function. Clearly the residue of the integrand at $s=2$ is
 $6\zeta(3)/\pi^2t^2$. The pole at $s=0$ is of second order. To
 find its residue we use the well-known expansions:
 $$
 \zeta(s+1)=\frac{1}{s}+C+...,
 $$
 $$
 \Gamma(s)=\frac{1}{s}-C+...
 $$
 $$
 \zeta(s-1)=\zeta(-1)+\zeta^\prime (-1) s+...
 $$
 $$
 t^{-s}=1-s\ln{t}+...
 $$
 and
 $$
 \frac{1}{\zeta(s)}=-2+2s\ln{(2\pi)}+... ,
 $$
 where $C$ denotes the Euler's constant. Multiplying these series,
 we obtain that the required residue is
 $-\frac{1}{6}\ln{t}-2\zeta^\prime(-1)-\frac{1}{6}\ln{(2\pi)}$.
 Finally, the residues at the zeta-zeros are accumulated by the
 sum
 \begin{equation} \label{zz}
 \sum_r t^{-r}\Gamma(r)\zeta(r+1)\zeta(r-1)=g(t),
 \end{equation}
 where $r$ runs over the non-trivial zeta-zeros
 . In this way we obtain
 \begin{eqnarray}
 & & \ln{S(e^{-t})} =\frac{6\zeta(3)}{\pi^2 t^2} +g(t)
 -\frac{1}{6}\ln{t} -2\zeta^\prime(-1) -\frac{1}{6}\ln{(2\pi)}  \nonumber \\
 & & -\frac{1}{2\pi i}\int_{-c_1-i\infty}^{-c_1+i\infty}
 t^{-s}\Gamma(s)\zeta(s+1)\frac{\zeta(s-1)}{\zeta(s)}ds. \nonumber
 \end{eqnarray}
 To estimate the last integral we use the following bounds for the
 zeta and gamma functions: $\zeta(1-c_1+iy)\ll \mid
 y\mid^{c_1/2}\ln{\mid y\mid}, \zeta(-c_1-1+iy)\ll\mid
 y\mid^{-1/2-c_1}\ln{\mid y\mid}, \Gamma(-c_1+iy)\ll\mid
 y\mid^{-c_1-1/2}e^{-\pi\mid y\mid/2}, \mid y\mid\ge y_0>0$ (see
 e.g. \cite{ivic2003}, p. 25 and p. 492). Finally, we estimate $1/\zeta(s)$
 using the well-known representation
 \begin{equation}\label{funceq}
 \zeta(s)=\chi(s)\zeta(1-s), \quad \chi(s)=(2\pi)^s/(2\Gamma(s)\cos{(\pi
 s/2)}).
 \end{equation}
  (see \cite[p.~9]{ivic2003}). For
 $s=-c_1\pm\mid y\mid$, Stirling's formula implies that
 $\mid\chi(-c_1\pm \mid y\mid)\mid=
 (\mid y\mid/(2\pi))^{c_1+1/2}(1+O(1/\mid y\mid), \mid y\mid\ge
 y_0>0$. Moreover, $\mid\zeta(1-(-c_1+i\mid y\mid)\mid
 =\mid\zeta(1+c_1+i\mid y\mid)\mid\ge (1-\epsilon)\zeta(1+c_1),
 0<\epsilon<1$, by \cite[Thm.\ 9.1, p.\ 235]{ivic2003}. This shows that
 $1/\zeta(-c_1+iy)\ll\mid y\mid^{c_1+1/2}$.
 Hence the integrand is absolutely integrable and the integral tends to $0$ as
 $t\to 0^+$ since
 \begin{eqnarray}
 & & \frac{1}{2\pi i}\int_{-c_1-i\infty}^{-c_1+i\infty}
 t^{-s}\Gamma(s)\zeta(s+1)\frac{\zeta(s-1)}{\zeta(s)}ds \nonumber
 \\
 & & \ll t^{c_1}\int_0^{\infty}e^{-\pi\mid y\mid/2} \mid
 y\mid^{-(c_1+1)/2}\ln^2{\mid y\mid}dy\ll t^{c_1}. \nonumber
 \end{eqnarray}
 Thus, we finally obtain
 \begin{equation}\label{lns}
 \ln{S(e^{-t})}=\frac{6\zeta(3)}{\pi^2 t^2} +g(t)
 -\frac{1}{6}\ln{t} -2\zeta^\prime(-1) -\frac{1}{6}\ln{(2\pi)}+o(1), \quad t\to
 0^+.
 \end{equation}
 The proposition now follows after changing the variable $t$ into
 $1-z$ and taking exponents from both sides of (\ref{lns}).

 We conclude the proof of Proposition 1 with an estimate for the
 function $g(t)$ as $t\to 0^+$. Using (\ref{funceq}), one can
 easily check that

 \begin{equation}\label{dupl}
 \zeta(s-1) =\frac{\zeta(2-s)}{2(2\pi)^{1-s}\sin{(\pi s/2)}
 \Gamma(s-1)}
 \end{equation}

   We represent $\zeta(s)$ in the denominator of (\ref{integrand})
 using Hadamard's factorization theorem \cite[p.~30-31]{Titchmarsh86} as
 follows:
 \begin{equation}\label{hadamard}
 \zeta(s)
 =\frac{e^{bs}}{2(s-1)\Gamma(s/2+1)}\prod_r\left(1-\frac{s}{r}\right)
 e^{s/r},
 \end{equation}
 where $b=\frac{1}{6}\ln{(2\pi)}-1-C/2$ and $C$ denotes Euler's constant. We
 assume first that the zero $r=\delta+i\gamma$ $(0\le\delta\le 1)$ of $\zeta(s)$ is simple
 (i.e. in (\ref{res}) $m_r=1$). Combining
 (\ref{integrand}), (\ref{dupl}) and (\ref{hadamard}), we obtain
 \begin{equation}\label{smrf}
 (s-r)F(s)
 =\frac{t^{-s}s^2(1-s)\Gamma(s/2)\zeta(s+1)\zeta(2-s)e^{-bs}}
 {2(2\pi)^{1-s}\sin{(\pi s/2)}\prod_{k\ge 1, r_k\neq r} (1-s/r_k)
 e^{s/r_k}}.
 \end{equation}
 In the last formula we assume that the complex zeros $r_k=\delta_k+i\gamma_k$ of
 $\zeta(s)$ are arranged in a non-decreasing order of the absolute
 values of their imaginary parts $\mid\gamma_k\mid$; if
 some absolute values coincide the order between them is taken
 arbitrarily. Clearly, $0\le\delta_k\le 1, k=1,2,...$. We shall
 study the behavior of $(s-r)F(s)$ in a neighborhood of $s=r$,
 say, $\{s:\mid s-r\mid<\epsilon\}$, where $\epsilon>0$ is small
 enough. Let us take now a positive number $T$ that satisfies the
 inequality $T+1/K_1\ln{T}<\mid\gamma\mid$, where $K_1>0$ will be
 specified later. Further on, by $K_1, K_2,...$ we shall denote
 some positive constants. We represent the product in the
 denominator of (\ref{smrf}) in the following way:
 $$
 \prod_{k\ge 1, r_k\neq r}(1-s/r_k)e^{s/r_k} =\Pi_1(s)\Pi_2(s),
 $$
 where
 $$
 \Pi_1(s)=\prod_{r_k\neq r,
 T\le\mid\gamma_k\mid<T+1}(1-s/r_k)e^{s/r_k}, \quad \Pi_2(s) =\prod_{r_k\neq r,
\mid\gamma_k\mid\notin[T,T+1)}(1-s/r_k)e^{s/r_k}.
$$
It is now clear that (\ref{smrf}) can be represented as follows:
\begin{equation}\label{pipsi}
(s-r)F(s)=\frac{\psi(s)}{\Pi_1(s)},
 \end{equation}
 where
\begin{equation}\label{psi}
 \psi(s)
=\frac{t^{-s}s^2(1-s)\Gamma(s/2)\zeta(s+1)\zeta(2-s)e^{-bs}}
 {2(2\pi)^{1-s}\sin{(\pi s/2)}\Pi_2(s)}
\end{equation}
is analytic in a neighborhood of $s=r$. Hence $\lim_{s\to
r}\psi(s)=\psi(r)$. To estimate $\psi(r)$ we first notice that
$\zeta(s)$ has no zeros with real parts equal to $0$ or $1$ (see
\cite[p.~49]{Titchmarsh86}). 
Hence both $\Re{(r+1)}, \Re{(2-r)}\in (1,2)$
and by \cite[Thm.\ 1.9, p.\ 25]{ivic2003},
$$
\zeta(r+1)\ll\ln{\mid\gamma\mid}, \quad
\zeta(2-r)\ll\ln{\mid\gamma\mid}.
$$
Furthermore
$$
\frac{1}{\mid\sin{(\pi r/2)}\mid}\ll e^{-\pi\mid\gamma\mid/2}, \quad
\Gamma(r/2)\ll e^{-\pi\mid\gamma\mid/4}
$$
and
$$
t^{-r}\ll t^{-\theta},
$$
where $\theta$ is the least upper bound for the real parts of the
zeta-zeros. Finally, $1/\mid\Pi_2(r)\mid$ is bounded away from $0$
since the product over those $r_k$ satisfying $\mid\gamma_k\mid\le
T$ has no zeros in the disc $\{s:\mid s-r\mid<\epsilon\}$ if
$\epsilon$ is enough small and the other factors for which
$\mid\gamma_k\mid>T+1$ are geometric progressions with ratios
equal to $\mid r/r_k\mid<1$. Combining these estimates with
(\ref{psi}), for $\mid s-r\mid<\epsilon$, we obtain
\begin{equation}\label{psiestimate}
\mid\psi(s)\mid\ll
t^{-\theta}\mid\gamma\mid^3\ln^2{\mid\gamma\mid}
e^{-3\pi\mid\gamma\mid/4}.
\end{equation}
Finally, we have to estimate $\lim_{s\to r} 1/\mid\Pi_1(s)\mid$.
We shall use some basic facts from the theory of the distribution
of the zeta-zeros in the critical strip. First, we apply the known
fact telling us that, for enough large $T$, the number of the zeta
zeros with absolute values of their imaginary parts laying in the
interval $[T,T+1)$ is at most $O(\ln{T})$ (see \cite[p.\ 211]{Titchmarsh86}). 
Let $K_1$ be the constant in this $O$-estimate. We also use
an estimate for the average spacings of two successive zeros of
$\zeta(s)$ in a neighborhood of $r=\delta+i\gamma$. It is $\sim
2\pi/\ln{\mid\gamma_k\mid}\sim 2\pi\ln{\mid\gamma\mid}$ (see
\cite[pp.\ 214, 246]{Titchmarsh86}). Finally, it is clear that $\mid
r_k\mid=\mid\delta_k+i\gamma_k\mid=O(T)=O(\mid\gamma\mid)$, for
$0\le\delta,\delta_k\le 1$ and $T\le \gamma,\gamma_k<T+1$ whenever
$T$ is enough large. Therefore, we have
\begin{eqnarray}
& & \mid\lim_{s\to r}\Pi_1(s)\mid
 =\prod_{r_k\neq r,T\le\mid\gamma_k\mid<T+1} \mid 1-\frac{r}{r_k}\mid
e^{\Re{(r-r_k)}} \ge K_2\prod_{r_k\neq r,
T\le\mid\gamma_k\mid<T+1} \frac{\mid
r-r_k\mid}{\mid r_k\mid} \nonumber \\
& & \ge
K_3\left(\frac{2\pi}{\ln{\mid\gamma\mid}}\right)^{K_1\ln{T}}
T^{-K_1\ln{T}} \ge
K_4\left(\frac{\pi}{\mid\gamma\mid\ln{\mid\gamma\mid}}\right)^{K_1\ln{\mid\gamma\mid}}
\ge K_5 e^{-K_6\ln^2{\mid\gamma\mid}}. \nonumber
\end{eqnarray}
Combining this estimate with (\ref{res}) (where $m_r=1$) and
(\ref{pipsi})-(\ref{psiestimate}), we obtain
\begin{eqnarray}\label{resest}
& & \mid Res_{s=r}(F(s)\mid\ll t^{-\theta}\mid\gamma\mid^3
\ln^2{\mid\gamma\mid}
\exp{(\ln^2{\mid\gamma\mid}-3\pi\mid\gamma\mid/4)} \nonumber \\
& & \ll t^{-\theta} e^{-K\mid\gamma\mid}, \quad K>0.
\end{eqnarray}
Although all known zeros $r$ of $\zeta(s)$ in the critical strip
are simple (i.e. $m_r=1$), it is not yet known whether this is in
fact true. At present, there are only estimates for $m_r$ (for
recent results in this direction, see \cite{Ivic1999}). The simplest
and oldest estimate seems to be $m_r\ll\ln{\mid\gamma\mid}$ (see
\cite[p.\ 211]{Titchmarsh86}). The proof of the fact that (\ref{resest})
is true whenever $m_r>1$ is technically more complicated. It
should follow the line of reasoning given above. One has to deal
with a sum containing $m_r=O(\ln{T})$ summands, for
$r=\delta+i\gamma, 0\le\delta\le 1, T\le\mid\gamma\mid<T+1$. Each
of these summands must contain after the differentiation the
factor $((s-r)/t)^{m_r-j}t^{-s}, j=0,1,...,m_r-1$. Since both
$s-r$ and $t$ tend to $0$, one should to keep a balance in each
summand, so that $((s-r)/t)^{m_r-j}\ll 1$. For the remaining
factors in each summand, estimates similar to those presented
above seem to be true. In this way (\ref{resest}) can be
established.

To obtain the final estimate for $g(t)$, we need to use the bound
\cite[p.\ 211]{Titchmarsh86}
$$
\sum_{r=\delta+i\gamma, T\le\mid\gamma\mid<T+1} 1\ll\ln{T}.
$$
Hence, by (\ref{sumres}),
$$
g(t)\ll\sum_{r=\delta+i\gamma} e^{-K\mid\gamma\mid/2} t^{-\theta}
\ll t^{-\theta}.
$$

%
%
\end{proof}
\begin{remark}Note that
 \begin{equation}\label{res}
 Res_{s=r}(F(s)) =\dfrac{1}{(m_r-1)!}\lim_{s\to r}\frac{d^{m_r-1}}{ds^{m_r-1}}((s-r)^{m_r}F(s)),
 \end{equation}
 where $m_r$ denotes the multiplicity of the zeta zero $r$. 
\end{remark}

\begin{remark} Let $\theta$ denote the least upper bound of the
real parts of non-trivial zeros of the Riemann zeta function
($\theta=1/2$ if Riemann hypothesis is true). We also proved that the function $g(t)$ defined above
satisfies $g(t)=O(t^{-\theta})$ as $t\to 0$.
\end{remark}

As a by-product, the equivalence in Proposition 1 allows us to
calculate, using technical approach \cite{Flajolet1991} following saddle point methods, the asymptotic growth of the number of NW-convex paths
of size $n$:

 \begin{proposition}
 For the number $p_{NW}(n)$ of NW-convex paths of size $n$, we have
 
 $$
p_{NW}(n) \sim\alpha n^{-11/18}
 \exp{\left(\beta
 n^{2/3}+g\left(\left(\frac{12\zeta(3)}{n\pi^2}\right)^{1/3}\right)\right)},
 $$
with $ g(t)=\sum_r t^{-r}\Gamma(r)\zeta(r+1)\zeta(r-1) $
where $r$ runs over the non-trivial zeros of the Riemann zeta function and

$$\alpha=\dfrac1{6}\,{\frac {{2}^{5/9}{{\rm e}^{{\frac {(5/2)\,\zeta  \left( 3 \right) -
2\,\zeta^\prime  \left(-1 \right) {\pi }^{2}}{{\pi }^{2}}}}}\sqrt [9]{
\zeta  \left( 3 \right) }{3}^{{\frac {11}{18}}}}{{\pi }^{{\frac {8}{9}
}}}}
\sim 0.3338488807...$$
 $$\mbox{and } \beta=\frac{3}{2^{2/3}}\left(\frac{\zeta(3)}{\zeta(2)}\right)^{1/3}= \dfrac{2^{-1/3} 3^{4/3}\zeta(3)^{1/3}}{\pi^{2/3}}\sim 1.702263426...$$
 \end{proposition}
\begin{proof}
The function $S(z)$ can be analyze exactly in the same vein that the generating function of integer partitions (see. \cite[p574-578]{Flajolet2009}). This explains that we can use saddle point analysis on the approximation of $S(z)$ to obtain the asymptotics of $p_{NW}(n).$ The aim of the proof is to verify the H-admissibility of the function $S(z)$. After that, the calculations are standard and the result ensues.

So, let $x_n=1-\left(\frac{12\zeta(3)}{\pi^2 n}\right)^{1/3}$ be the
root of the saddle point equation and let
$t_n=-\ln{x_n}\sim\left(\frac{12\zeta(3)}{\pi^2 n}\right)^{1/3}$
(i.e., $x_n=e^{-t_n}$). We shall prove the following decay property:
\begin{equation}\label{decay}
\mid S(x_n e^{i\theta})\mid=o(S(x_n)/\sqrt{b(x_n)}), \quad
n\to\infty,
\end{equation}
uniformly for $t_n\le\mid\theta\mid<\pi$.

 First we notice that
 \begin{equation} \label{ratio}
 \frac{\mid
S(e^{-t_n+i\theta})\mid}{S(e^{-t_n})}
=\exp{(\Re{(\ln{S(e^{-t_n+i\theta}}))}-\ln{S(e^{-t_n}))}}.
\end{equation}
 (Here $\ln{(.)}$ denotes the main branch of the
logarithmic function, so that $\ln{y}<0$ if $0<y<1$.) Then,
setting $\theta=2\pi u$, for $t_n/2\pi\le\mid\theta\mid/2\pi=\mid
u\mid<1/2$, we have
\begin{eqnarray}\label{long}
& & \Re{(\ln{S(e^{-t_n+i\theta}})}-\ln{S(e^{-t_n}))} = \nonumber
\\
& & =\Re{\left(-\sum_{k=1}^\infty\varphi(k)
\ln{\left(\frac{1-e^{-kt_n+2\pi i uk}}{1-e^{-k
t_n}}\right)}\right)} \nonumber \\
& & =-\frac{1}{2}\sum_{k=1}^\infty\varphi(k)
\ln{\left(\frac{1-2e^{kt_n}\cos{(2\pi uk)}+e^{-kt_n}}
{(1-e^{kt_n})^2}\right)} \nonumber \\
& & =-\frac{1}{2}\sum_{k=1}^\infty\varphi(k) \ln{\left(1+
\frac{4e^{kt_n}\sin^2{(\pi uk)}}{(1-e^{kt_n})^2}\right)}
\nonumber \\
& & \le -\frac{1}{2}\sum_{k=1}^\infty\varphi(k)
\ln{(1+4e^{kt_n}\sin^2{(\pi uk)})} \nonumber \\
& & \le -\frac{\ln{5}}{2}\sum_{k=1}^\infty\varphi(k) e^{-kt_n}
\sin^2{(\pi uk)} =-\frac{\ln{5}}{2} U_n,
\end{eqnarray}
where the last inequality follows from the fact that
$\ln{(1+y)}\ge\frac{\ln{5}}{4}$ for $0\le y\le 4$. Further, we
shall obtain a bound from below for the sum
\begin{equation}\label{un}
U_n=\sum_{k=1}^\infty\varphi(k) e^{-kt_n} \sin^2{(\pi uk)}.
\end{equation}
We also denote by $\{d\}$ the fractional part of the
 real number $d$, and by $\parallel d\parallel$ the
 distance from $d$ to the nearest integer, so that
 $$
 \parallel d\parallel=\left\{\begin{array}{ll} \{d\} & \qquad  \mbox {if}\qquad \{d\}\le 1/2, \\
 1-\{d\} & \qquad \mbox {if}\qquad \{d\}>1/2.
 \end{array}\right.
 $$
 It is not difficult to show that
 \begin{equation} \label{sinlowbound}
 \sin^2{(\pi d)}\ge 4\parallel d\parallel^2
 \end{equation}

 We notice that $\parallel uk\parallel=uk$ if $\mid u\mid k<1/2$,
 that is, if $k<1/2\mid u\mid\le\pi/t_n$. Hence, applying
 (\ref{sinlowbound}), the inequality $\mid u\mid\ge t_n/2\pi$ and the fact that $\varphi(k)\ge c_0
 k/\ln{\ln{k}}$ for $k\ge 3$ (this fact is due to Landau), we obtain
 \begin{eqnarray}\label{unest}
 & & U_n\ge 4\sum_{k=1}^\infty \varphi(k) e^{-kt_n}\parallel
 uk\parallel^2 \nonumber \\
 & & \ge 4u^2\sum_{1\le k\le\pi/t_n} k^2\varphi(k) e^{kt_n}
 \ge\frac{t_n^2}{\pi^2} \sum_{3\le k\le\pi/t_n}
 k^2\left(c_0\frac{k}{\ln{\ln{k}}}\right)e^{kt_n} \nonumber \\
 & & \ge \frac{c_0 t_n^2}{\ln{\ln{(\pi/t_n)}}}t_n^{-4}\int_0^\pi y^3
 e^{-y}dy \ge c_1\frac{t_n^{-2}}{\ln{\mid\ln{t_n}\mid}}.
 \end{eqnarray}
 Here $c_0$ and $c_1$ are positive constants. Replacing
 (\ref{long}), (\ref{un}) and (\ref{unest}) into (\ref{ratio}) and
 taking into account the asymptotic equivalence for $t_n$, we get
 $$
 \mid S(x_n e^{i\theta})\mid =O(S(x_n) \exp{(-c_3
 n^{2/3}/\ln{\ln{n}})}).
 $$
This implies immediately (\ref{decay}) since $\sqrt{b(x_n)}$ is of
order $const.n^{2/3}$ - much smaller than the exponential one
given above.

\end{proof}
\begin{remark}
The contribution of $g\left(\left(\frac{12\zeta(3)}{n\pi^2}\right)^{1/3}\right)$ is a fluctuation of very small amplitude as it is classically observe in similar analysis. 
In particular, this contribution is imperceptible on the first 1000 coefficients.
\end{remark}

Now, we focus on the study of the average number of initial vertical steps 
(which corresponds to the size of the first block of 1 in its associated word) in a NW-convex path. 
\begin{lemma}\label{init}
 The average number of initial steps is equivalent to ${\frac {\sqrt [3]{18{\pi }^{2}n}}{6\sqrt [3]{
\zeta  \left( 3 \right) }}}.$
\end{lemma}
\begin{proof}
 A classical way to tackle this type of problem is to mark by a new parameter the contribution of initial steps in the generating function. So, we have the following bivariate generation function $S(z,u)=(1-zu)^{-1}\prod\limits_{n=2}^{\infty}
(1-z^{n})^{-\varphi(n)}$ where clearly, the coefficient $[z^nu^k]S(z,u)$ is the number of NW-convex path of size $n$ having exactly $k$ initial vertical steps. Now, the average number of initial steps for the NW-convex path of size $n$ is just $\dfrac{[z^n]\dfrac{\partial S(z,u)}{\partial u}|_{u=1}}{[z^n]S(z,1)}$. 
So, we need to extract the asymptotic of $G(z):=\dfrac{\partial S(z,u)}{\partial u}|_{u=1}=\dfrac{z}{(1-z)^{2}}\prod\limits_{n=2}^{\infty}(1-z^{n})^{-\varphi(n)}$. Again, we proceed by a Mellin transform approach and we get that:
 $$\cM[\ln(G(e^{-t}))](s)={\frac {\zeta  \left( s+1 \right) (\zeta  \left( s-1 \right)+ \zeta  \left( s \right))\Gamma
 \left( s \right) }{\zeta  \left( s \right) }}$$
 So, $G(z)\sim \dfrac{\exp\left({\frac {6z\zeta  \left( 3 \right) }{{\pi }^{2} \left( 1-z \right) ^{2
}}}\right)}{ \left(2\pi(1-z)^7\right)^{\frac{1}{6}}}\cdot \exp\left(
{\dfrac {g(1-z)+\zeta  \left( 3 \right) }{2{\pi
}^{2}}}-2\,\zeta^\prime \left(-1 \right)\right ).$
Finally, using saddle point analysis and dividing by the asymptotic of $p_{NW}(n)$, we obtain lemma \ref{init}.

\end{proof}

In particular, if we renormalize the NW-convex path by $1/n$, the contribution of the initial steps for the limit shape is null.

Now, we are interested in the average position of the terminating point of a random NW-convex path.
 If we consider NW-convex path without their initial vertical steps, then by symmetry,
 we can conclude that the average ending position is $(x_n\sim\frac{n}{2},y_n\sim\frac{n}{2})$. 
But by lemma~\ref{init} and the fact that the length of the renormalized initial vertical steps is $o(1)$, it follows that:
\begin{lemma}\label{end}
 The average position of the ending point of a random NW-convex path of size $n$ is $(x_n\sim\frac{n}{2},y_n\sim\frac{n}{2}).$
\end{lemma} 

Following the same approach, with a little more work, we can prove that:

\begin{proposition}
\label{prop:abs}
The average abscissa of the point of slope $x$ in a renormalized by $1/n$ 
NW-convex path of size $n$ is $\dfrac1{2}\left(1-\left(\dfrac{x}{1+x}\right)^2\right).$
\end{proposition}
\begin{proof}
From lemma \ref{end} we deduce:
$$\dfrac{\partial (1-zv)S(z,h,v)}{\partial h}|_{h=1,v=1}=\dfrac{z}{2}\dfrac{\partial (1-zv)S(z,h,v)}{\partial z}|_{h=1,v=1},$$
Additionally, expanding of the derivative, we have 
$$\dfrac{\partial (1-zv)S(z,h,v)}{\partial h}|_{h=1,v=1}=(1-zv)S(z,h,v)|_{h=1,v=1}\sum\limits_{n=2}^{\infty}\sum\limits_{p+q=n, p\wedge q=1}\dfrac{pz^n}{1-z^n}.$$

So, we deduce from this the following technical lemma which will be useful in the sequel:
\begin{lemma}\label{useful} The following identity holds:
$$\sum\limits_{n=2}^{\infty}\dfrac{z^n}{1-z^n}\sum\limits_{p+q=n, p\wedge q=1}p=\sum\limits_{n=2}^{\infty}\dfrac{n\varphi(n)z^n}{2(1-z^n)}=\dfrac{z}{2}\dfrac{\dfrac{\partial (1-zv)S(z,h,v)}{\partial h}|_{h=1,v=1}}{(1-zv)S(z,h,v)|_{h=1,v=1}}.$$
\end{lemma}

Let us continue with the average position of the point of slope $x$ for $0\leq x <\infty$. The position of the ending point, previously done, is a special case of this question (where $x=0$).
For that purpose, let us consider the generating function $$F_x(z,u)=(1-z)^{-1}\prod\limits_{n=2}^{\infty}\prod\limits_{p+q=n, p\wedge q=1, \frac{q}{p}>x}(1-z^{n}u^p)^{-1}\prod\limits_{p+q=n, p\wedge q=1, \frac{q}{p}\leq x}(1-z^{n})^{-1}.$$ The average abscissa of the point of slope $x$ in a NW-convex path of size $n$ is nothing but the expectation $\dfrac{[z^n]\dfrac{\partial F_x(z,u)}{\partial u}|_{u=1}}{[z^n]F_x(z,1)}.$ Expanding the derivative, we get:
$$\dfrac{\partial F_x(z,u)}{\partial u}|_{u=1}=F_x(z,u)|_{u=1}\sum\limits_{n=2}^{\infty}\dfrac{z^n}{1-z^n}\sum\limits_{p+q=n, p\wedge q=1,\frac{q}{p}>x}p.$$
Now, we need the following short number theory lemma:
\begin{lemma} The following equivalence holds for every fixed $x\in [0,\infty[$:\\ 
$\sum\limits_{p+q=n, p\wedge q=1,\frac{q}{p}>x}p=\dfrac1{2}\left(1-\left(\dfrac{x}{1+x}\right)^2\right)n\varphi(n)(1+\varepsilon(n))$ with $\varepsilon(n)\rightarrow 0$ as $n$ tends to $\infty$.
\end{lemma}
The proof is quite immediate. 

Finally, using the lemma~\ref{useful}, we get proposition \ref{prop:abs}.
\end{proof}

At this stage, to prove that the limit shape is deterministic, we need to show that the standard deviation of the abscissa is in $o(n)$. 
This proof is long and technical, but follows the same way that we do for the expectation.
We conclude by solving the differential equation:
 $\left\{ f \left( 0 \right) =0,2f \left( z \right) =1-{\frac {
 \left( {\frac {d}{dz}}f \left( z \right)  \right) ^{2}}{ \left( 1+{
\frac {d}{dz}}f \left( z \right)  \right) ^{2}}} \right\} 
$ which explains the fact that the slope of $f(z)$ is $x$ at the abscissa $\dfrac1{2}\left(1-\left(\dfrac{x}{1+x}\right)^2\right).$ 
Consequently, we have:

\begin{theorem}
The limit shape for the renormalized by $1/n$ NW-convex path of size $n$ as $n$ tends to the infinity is the curve of equation $f(z) = \sqrt{2z}-z$.
\end{theorem}

\section{Boltzmann samplers for NW-convex paths}
\label{section3}

Boltzmann samplers have been introduced in 2004 by Duchon \textit{et
  al.} \cite{DuFlLoSc04} as a general framework to generate uniformly at random
specifiable combinatorial structures. These samplers are very
efficient, their expected complexity is typically linear. 
In comparison with the recursive method of sampling, the
principle of Boltzmann samplers essentially deals with evaluations
of the generating function of the structures, and avoids counting coefficient (which need to be pre-computed in the recursive
method). Quite a few papers have been written to extend and
optimize Boltzmann samplers
\cite{PiSaSo08,BoJa08,BoRoSo,BoPo10,BoFuKaVi,BoGaRo10,BoLu12}.

Consequently, we choose Boltzmann sampling for the class of digitally
convex polyominoes. After a short introduction to the method, we
analyze the complexity of the sampler. This part is based on an
analysis by Mellin transform techniques.

\subsection{A short introduction to Boltzmann samplers}
Let us recall the definitions and the main ideas of Boltzmann sampling.
\begin{definition}
Let $\cA$ be a combinatorial class and $A(x)$ its ordinary
generating function. A \emph{(free) Boltzmann sampler} with
parameter $x$ for the class $\cA$ is a random algorithm
$\Gamma_x\cA$ that draws each object $\gamma \in \cA$ with
probability:
\[\mathbb{P}_x (\gamma )=\frac{x^{|\gamma |}}{A(x)}.\]
Notice that this definition is consistent only if $x$ is a
positive real number taken within the disk of convergence of the
series $A(x)$.
\end{definition}

The great advantage of choosing the Boltzmann distribution for the
output size is to allow simple and automatic rules to design Boltzmann
samplers from a specification of the class. 

Note that from free Boltzmann samplers, we easily define two
variants: the \emph{exact-size Boltzmann sampler} and the
\emph{approximate-size} one, just by rejecting the inappropriate
outputs until we obtain respectively a targeted size or a targeted
interval of type $[(1-\eps ) n,(1+\eps ) n]$. In order to optimize
this rejection phase, it is crucial to tune efficiently the
parameter $x$. A good choice is generally to take the unique
positive real solution $x_n$ (or an approximation of it when $n$
tends to infinity) of the equation $\dfrac{xA'(x)}{A(x)}=n$ which
means that the expected size of the output tuned by $x_n$
equals $n$.

To conclude, let us recall that authors of the seminal paper \cite{DuFlLoSc04}
distinguished a special case where Boltzmann samplers are particularly
efficient (and we prove in the sequel, that we are in this
situation). This case arises when the Boltzmann distribution of the
output is \emph{bumpy}, that is to say, when the following conditions
are satisfied:
\vspace{-0.3cm}
\begin{itemize}
\item $\mu_1(x) \rightarrow \infty$ when $x \rightarrow \rho^{-}$
\item $\sigma(x)/\mu_1(x) \rightarrow 0$ when $x \rightarrow \rho^{-},$
\end{itemize}
where $\mu_1(x)$ (resp.  $\sigma^2(x)$) is the expected size (resp. the variance) of the output, and $\rho$ is the dominant singularity of $A(x)$.
\subsection{The class of the digitally convex polyominoes}
Let us recall that the digitally convex polyominoes can be
decomposed into four NW-convex paths, each of them being a multiset of
discrete irreducible segments. Moreover, according to the previous
section, we have a specification for the discrete irreducible
segments in terms of word theory. This brings us the generating
function associated to a NW-convex path:
\[S(z)=\prod\limits_{n=1}^{\infty} (1-z^{n})^{-\varphi(n)},\]
where $\varphi(n)$ designs the Euler's totient function.

The first question that occurs concerns the determination of the
parameter $x_n$ which is a central point for the tuning of the
sampler. In order to approximate $x_n$ as $n$ tends to infinity,
we need to apply the asymptotic of $S(z)$ as $z \rightarrow 1$, we already calculate for the asymptotic of the NW-convex paths.

\subsection{The Boltzmann distribution of the NW-convex paths}
The first step to analyze the complexity of exact- and
approximate-size Boltzmann sampler is to characterize the type of
the Boltzmann distribution. In this subsection we prove that the Boltzmann distribution
is bumpy. This ensures that we only need on average a constant
number of trials to draw an object of approximate-size. Moreover,
a precise analysis allows us to give the complexity of the
exact-size sampling.

Firstly, we derive from the equivalence of $S(z)$ close to its dominant singularity $\rho=1$, an expression for the tuning of the Boltzmann
parameter:

\begin{corollary}
\label{corgoodchoice}
A good choice for the Boltzmann parameter $x_n$ in
order to draw a large NW-convex paths of size $n$ is
$x_n=1-\sqrt [3]{{\dfrac {12\zeta  \left( 3 \right) }{n{\pi }^{2}
}}}.$

\end{corollary}
\begin{proof}
The expected size of the output is $\dfrac{xS'(x)}{S(x)}$ which is an increasing function in $x$. Using the equivalent of $S(x)$ when $x$ tends to 1, we can approximate the first member of the equation $\dfrac{xS'(x)}{S(x)}=n$ to obtain $\dfrac {12\zeta  \left( 3 \right) }{ \left( 1-x \right) ^{3}{\pi }^{
2}}=n.$ The result ensues immediately.
\end{proof}

So, the first
condition for the bumpy distribution is clearly verified. We now
focus our attention on the fraction $\sigma(x)/\mu_1(x)$.

\begin{lemma} \label{lem:bumpy}
The expected size of the Boltzmann output satisfies:
$$\mu_1(x)\sim \dfrac {12\zeta  \left( 3 \right) }{ \left( 1-x
\right) ^{3}{\pi }^{ 2}} \mbox{ as } x \mbox{ tends to 1.}$$ The
variance of the size of the Boltzmann output satisfies:
$$\sigma(x)\sim \frac{6\,\sqrt {\zeta  \left( 3 \right)  {
{x}}}}{\pi{ \left( 1-x
 \right) ^{2}} }
\mbox{ as } x \mbox{ tends to 1.}$$
So, the Boltzmann distribution of the NW-convex paths is bumpy.
\end{lemma}


Finally, we
describe the sampler for digitally convex polyominoes. This needs
some care during the stage when we glue the NW-convex paths together.

\subsection{Random sampler to draw Christoffel primitive words}
We now look more precisely how to implement the samplers.
 Firstly,
we address the question  of drawing two coprime integers which
is non-classical in Boltzmann sampling, from which we derive a Boltzmann sampler for NW-convex paths.
The first step to generate NW-convex paths is to draw Christoffel primitive words with Boltzmann probability.
We recall that this is equivalent to draw two coprime integers $p,q$ with probability $\frac{x^{p+q}}{\sum_{n\geq 1}\varphi(n) x^n}$.

Let $b(x,n):=\dfrac{\varphi (n)
x^{n}}{\sum\limits_{n=1}^{\infty}{\varphi (n) x^n}}$. The
following algorithm is an elementary way to answer the question
posed above:

\begin{algorithm}[H]
\caption{$\Gamma_{CP} $}
\label{gammaCP}
\KwIn{a parameter $x$}
\KwOut{Two coprime integers}

Draw $n$ with Boltzmann probability $b(x,n)$\\

\textbf{Do} Draw $p$ uniformly in $\{1,...,n\}$\\
\textbf{While} $p, n$ not coprime\\
\Return $(p, n-p)$\\
\end{algorithm}
The average complexity of the algorithm is in $O(n\ln \ln
(n))$. 

 \subsubsection{Correctness}
$p \wedge n =1 \Leftrightarrow p \wedge n-p=1$, so drawing $p$ coprime with $n$ uniformly in $\{1,..,n\}$ is equivalent to draw uniformly $p,q$ coprime with $p+q=n$.
With the choice of $n$ with the probability $\frac{\varphi (n) x^{n}}{\sum\limits_{n=1}^{\infty}{\varphi (n) x^n}}$, we obtain a Boltzmann sampler for Christoffel primitive words.
\subsubsection{Complexity}
To evaluate the complexity of this sampler there are two steps
of the algorithm that should be analyzed: the drawing of $n$
with probability $b(x,n)$ and the drawing of two coprime numbers.
The drawing of $n$ is essentially related to the generating
function evaluation and to $\varphi(k)$ for all $k\leq n$. In
the following, we consider that both the generating function and a
table for large enough $\varphi(n)$ are to be precomputed. So,
the complexity of the drawing of $n$ is in $O(n)$. Experimentally,
the precomputation takes a couple of minutes on a standard
personal computer. Now let us evaluate the complexity of drawing
the coprime numbers. The probability for $p$ to be taken
uniformly from $\{1,..,n\}$ being also coprime with $n$ is
$\varphi(n)/n$
and the following classical inequality \cite[Thm 8.8.7]{BachShallit} on totient's function:\\
 $\dfrac{\varphi(n)}{n}>\dfrac1{e^\gamma\log \log
(n)+\dfrac{3}{\ln \ln (n)}}$ (valid for all positive integers
$n$) proves
that on average the number of trials in the loop is $O(\log \log (n))$.\\
\begin{remark}
In fact, we can expect a better complexity. Indeed, with Boltzmann probability, we have better chance to draw an integer $n$ such that $\varphi(n)$ is close to $n$ and in this case the probability to draw a coprime $p$ is positive.
\end{remark}

%
%
%
%


\subsection{Random sampler drawing a NW-convex path}

To draw a NW-convex path, we use the isomorphism between NW-convex paths and multisets of Christoffel primitive words. 
The multiset is a classical constructor, for which its Boltzmann sampler in the unlabelled case had been given in \cite{FlaFuPi07}.
The idea is to draw with an appropriate distribution (called IndexMax distribution) an integer $M$, 
and then draw a random number of Christoffel primitive words with a Boltzmann sampler of parameter $x^i$ and to 
replicate each drawn object $i$ times, for all $1\leq i\leq M$. Well chosen probabilities ensures the Boltzmann distribution.

Once we get a multiset of pairs of coprime integers, we can transform it into a NW-convex path coded on $\{0,1\}$ as follows:
\begin{itemize}
\item Draw a multiset $m$ in $\MSet(PC)$,
\item Sort the elements $(p,q)$ of $m$ in decreasing order according to $q/p$,
\item Transform each element into the discrete line of slope $q/p$ coded on $\{0,1\}$,
\item add a $1$ at the beginning.
\end{itemize}
Clearly, this transformation has a complexity in $O(n\ln(n))$, due to the sorting.

\subsection{Complexity of sampling a NW-convex path in approximate and exact size}

The previous sections bring all needed elements to determine the complexity of the Boltzmann sampler for a NW-convex path. 
The two following theorems respectively tackle the complexity of the sampling in the case of approximate-size output or exact size output.

\begin{theorem}
An approximate-size Boltzmann sampler for NW-convex paths succeeds in one trial with probability tending to 1 as $n$ tends to infinity. 
The overall cost of approximate-size sampling is $O(n\ln(n))$ on average.
\end{theorem}


\begin{theorem}
\label{theo3}
An exact-size Boltzmann sampler for NW-convex paths succeeds in mean number of trials tending to $\kappa\cdot n^{2/3}$
 with $\kappa={\dfrac {\sqrt [6]{2}\sqrt [3]{3}{\pi }^{5/6}}{\sqrt [6]{\zeta  \left(
3 \right) }}}\approx 4.075517917...$
as $n$ tends to infinity. The overall cost of exact-size sampling is $O(n^{5/3}\ln(\ln(n)))$ on average.
\end{theorem}

\begin{remark}
Since $\varphi(n)$ grows slowly, the parameter $x_n$ tuned to draw large objects will be close to $1$, which gives big replication orders.
A 
 consequence
is that we do not need to calculate the generating function of primitive Christoffel words to a huge order to have a good approximation of our probabilities.
\end{remark}
\subsection{Random sampler drawing digitally convex polyominoes}
We can now sample independent NW-convex paths with Boltzmann probability.
We want to obtain an entire polyomino by gluing four (rotated) NW-convex path. \\
However, gluing a 4-tuple of NW-convex paths, we do not necessarily obtain the contour of a polyomino.
 Indeed, we need the following extra conditions: the four NW-convex paths should be non-crossing, and they need to form a closed walk with no half-turn. \\

\begin{algorithm}[!h]
\caption{$\Gamma_{P} $}
\label{gammaP}
\KwIn{a parameter $x$}
\KwOut{a quadruple of compatible NW-convex paths.}
\textbf{Repeat}:\\
\hspace{0.6cm}Draw WN, NE, ES, SW using independent calls to a Boltzmann sampler of\\
\hspace{0.6cm}NW-convex path of parameter $x$.\\
\hspace{0.6cm}\textbf{If} $|WN|_0 + |NE|_1=|ES|_0+|SW|_1$ and $|NE|_0+|ES|_1=|SW|_0+|WN|_1$ \\
\hspace{.6cm}\textbf{then} \Return $(WN, NE, ES, SW)$\\
\end{algorithm}

The non-crossing and no half-turn conditions are trivially satisfied when the four paths are NW-convex. 
Then the closing problem stays and we need to add a rejection phase at this step.
More precisely, to be closed, we need to have as much horizontal steps in the upper part from W to E as
 in the lower part from E to W, and as much vertical steps in the left part from S to N as in the right part from N to S. 
 A naive way to sample DCP according to a Boltzmann distribution is presented in Algorithm \ref{gammaP}. A more efficient uniform sampler can probably be adapted from \cite{BoGaRo10} and is currently under development.

\begin{figure}[!h]

\hspace{-0.45cm}
\includegraphics[scale=0.5]{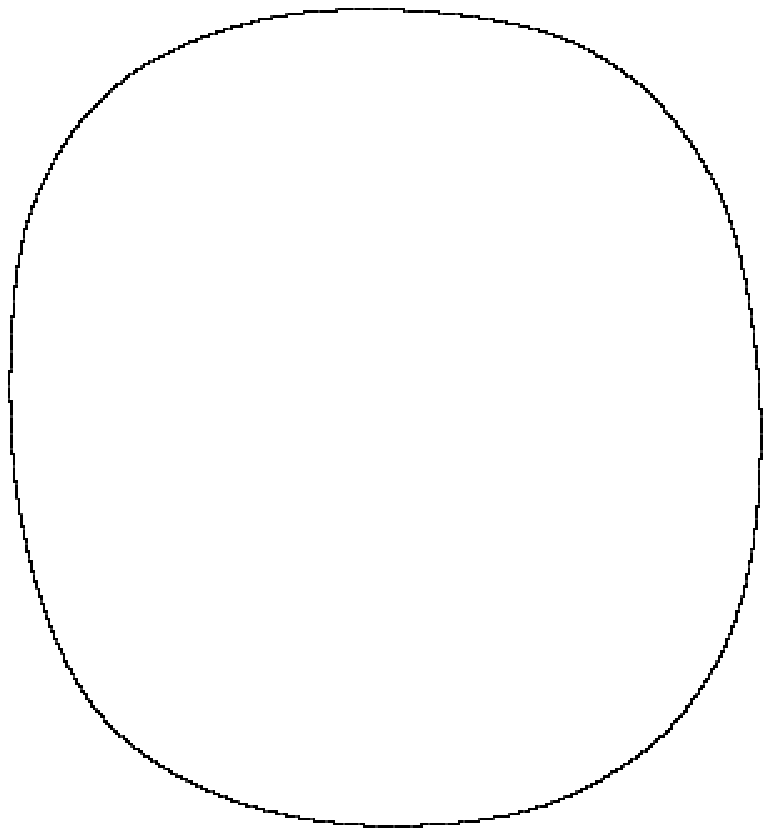}
\includegraphics[scale=0.32]{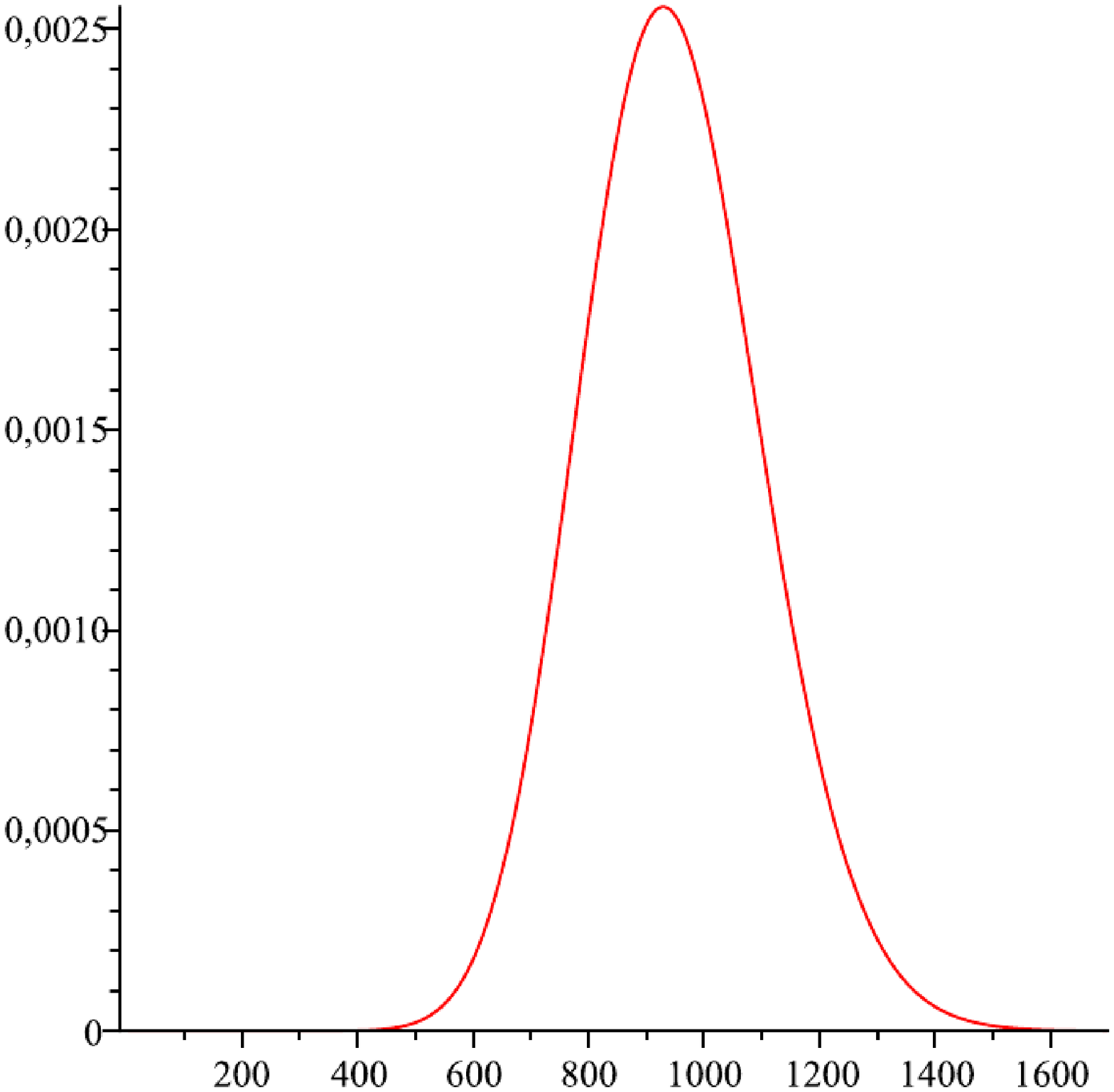}
\hspace{-0.35cm}
\caption{To the left, a random polyomino of perimeter 81109 drawn with parameter $x=0.98$. 
To the right perimeter distribution of a NW-convex path drawn with $x=0.8908086616$.}
\label{grosdessin}
\end{figure}
\section*{Conclusion}
We proposed in this paper an effective way to draw uniformly at random
digitally convex polyominoes. Our approach is based on Boltzmann
generators which allows us to build large digitally convex polyomioes.
These samples
point out the fact that random digitally convex polyominoes admit a
limit shape as their size tends to infinity. The limit shape of the NW-convex paths we proved in this paper seems to be also 
the limit shape for NW part of the digitally convex polyominoes. The tools to
tackle it are for the moment beyond our reach. 
Even, the simpler question of a precise asymptotic enumeration of the digitally convex polyominoes
 (the order of magnitude is proven \cite{IvKoZu94}) is currently a challenge. We conclude by noting that our
work could certainly be extended to higher dimensions. But, this is a
work ahead...
\subsection*{Acknowledgement}
The fourth author established the collaboration on this work during his visit at
LIPN, University of Paris 13. He wishes to thank for the fruitful discussions, hospitality and
support.
We thank Axel Bacher for its English corrections.
Olivier Bodini is supported by ANR project MAGNUM, ANR 2010 BLAN 0204.

\newpage

\bibliographystyle{plain}
\bibliography{digitconvsept}

\end{document}